\documentclass[12pt]{article}
\usepackage{amssymb}
\usepackage{authblk}
\usepackage{amsmath}
\usepackage{bm}
\usepackage{latexsym}
\usepackage{mathrsfs}
\usepackage{amsthm}
\usepackage{graphics}
\usepackage{graphicx}
\usepackage{xcolor}
\usepackage{overpic}
\newtheorem{prop}{Proposition}

\usepackage{hyperref}
\begin{document}
\title{\bf   $T\bar{T}$ deformation and multiple-flavor Lorentzian threads}
\author{Mojtaba Shahbazi\thanks{Corresponding author: mojtaba.shahbazi@modares.ac.ir}}
\author{Mehdi Sadeghi\thanks{mehdi.sadeghi@abru.ac.ir}}
\affil{Department of Physics, Faculty of Basic Sciences, Ayatollah Boroujerdi University, Boroujerd, Iran}
\date{\today}
\maketitle

\abstract{This work is motivated by the proposed relationship among finite cut-off holography and generalized $T\bar T$ deformations, and examines holographic complexity within the framework of the "complexity = anything" proposal. Employing a Fefferman-Graham expansion near the finite cut-off surface, the deformation-induced correction to generalized complexity is derived and shown to allow a systematic expansion in terms of generalized Willmore-type functionals. The resulting formulation broadens earlier findings for the complexity-volume proposal to encompass arbitrary geometric complexity measures. Furthermore, the structure of the correction allows a natural interpretation as multiple-flavor Lorentzian threads, where distinct thread sectors correspond to different curvature invariants in the complexity functional. These results show a geometric connection among finite cut-off holography, generalized complexity, and the emergence of non-local computational structures in holographic quantum field theories.}\\
\section{Introduction} \label{intro}
Irrelevant deformations of quantum field theories offer a solid framework for exploring non-local dynamics while preserving analytic control over physical observables. Notably, the $T\bar T$ deformation and its generalizations have garnered significant attention due to their integrability properties \cite{deform,deform1,deform2} and their holographic interpretation as finite cut-off geometries \cite{fc,fc1}. From a holographic standpoint, these deformations provide a concrete context in which modifications to the ultraviolet structure of the boundary theory are represented geometrically within the bulk spacetime.

Given that $T\bar T$-type deformations alter the structure of the boundary stress tensor and introduce non-local interactions, it is pertinent to investigate how these modifications are manifested in the complexity of the dual quantum state.

Holographic complexity has become a significant tool for probing quantum information in gravitational systems. Computational complexity, defined as the minimum number of quantum gates required to transform an initial state into a specified final state, has been incorporated into holography through several proposals. The "complexity=volume" (CV) conjecture identifies the complexity of the boundary theory with the volume of Einstein-Rosen bridges \cite{cv}. The "complexity=action" (CA) proposal relates complexity to the action of the bulk theory restricted to the Wheeler-DeWitt patch \cite{ca}, while "complexity=volume 2.0" (CV2.0) associates it with the spacetime volume within the Wheeler-DeWitt patch \cite{cv2}. More recently, the "complexity=anything" framework generalizes complexity to the volume of an arbitrary functional \cite{general}, providing a broader context for examining the geometric origins of complexity and the influence of various bulk structures on boundary state preparation. In this work, we focus on cases involving codimension-one bulk regions.

\begin{align}\label{obser}
\mathcal{O}_{F_1,\Sigma_{F_2}}(\Sigma_{CFT})=\frac{1}{G_N L}\int_{\Sigma_{F_2}}\mathrm{d}^{d}\sigma \sqrt{h} F_1(g_{\mu\nu};X^{\mu}),
\end{align}

where $F_1$ can equal one, which leads to the volume of the hypersurface as in CV. It can be general scalar function of metric $g_{\mu\nu}$ and an embedding $X^{\mu}(\sigma^a)$ of the hypersurfaces. Furthermore, $\Sigma_{F_2}$ is a codimension-one hypersurface in the bulk spacetime with boundary time slice $\partial\Sigma_{F_2} =\Sigma_{CFT}$. Extremality of the hypersurface leads to

\begin{align}\label{var}
\delta_{X}\Big(\int_{\Sigma}\mathrm{d}^{d}\sigma \sqrt{h} F_2(g_{\mu\nu};X^{\mu})\Big)=0.
\end{align}
For the sake of simplicity, in this work we follow the case $F_1 =F_2$, so the observable \eqref{obser}, known as $\mathcal{C}_{Any}$, obeying the above condition is expressed by:

\begin{align}\label{maxv}
\mathcal{C}_{Any}(\tau)= \max_{\partial\Sigma(\tau)=\Sigma_{\tau}}\frac{V_x}{G_N L}\left[\int_{\Sigma}\mathrm{d}^{d}\sigma \sqrt{h} F_1(g_{\mu\nu};X^{\mu})\right],
\end{align}

where $h$ is the determinant of the induced metric on the given hypersurface.

The usual choice of generalization is

\begin{align}\label{F1F2}
F_1=F_2=1+\alpha L^4 C^2,
\end{align}
where $C^2$ is the Weyl tensor squared and for $\alpha=0$ the CV is recovered. However, the generalized function $F_1$ admits any bulk curvature invariant.

At the same time, recent investigations of generalized complexity have revealed a richer structure than previously anticipated. In particular, transitions in the complexity growth rate appear to be closely related to properties of the boundary energy-momentum tensor. These observations suggest that generalized complexity may encode information about distinct computational sectors of the underlying quantum circuit. Understanding how such structures are modified in the presence of non-local deformations is therefore of particular interest.

In this work, we study the effect of $T\bar T$-induced non-locality on generalized holographic complexity. Using a Fefferman-Graham (FG) expansion near a finite radial cut-off, we derive the corrections to the generalized complexity functional generated by the deformation. We show that the resulting expression can be organized as a series of generalized Willmore-type contributions whose coefficients are determined by the geometric invariants entering the complexity functional, where, for the case of CV and CA, it has been evaluated in \cite{faraj}.

We then investigate the Lorentzian-thread interpretation of these corrections. We argue that generalized complexity naturally admits a decomposition into multiple thread sectors, each associated with a distinct curvature invariant. From this perspective, the appearance of multiple flavors of Lorentzian threads may be viewed as a geometric manifestation of the non-local computational structure induced by the $T\bar T$ deformation. The Lorentzian thread reformulation has been considered in CV \cite{pedraza2,pedraza3}, complexity-anything \cite{comlor}, and the notion of multiple Lorentzian threads in CV \cite{pedraza1}; however, these thread reformulations have been considered in the dual convex program. In this work, we begin with complexity=anything and construct multiple-thread reformulations across different geometric sectors; i.e., we show that there exists a local construction for every positive curvature invariant and a global construction under some conditions.

The paper is organized as follows. In section \ref{sec2}, we review $T\bar T$ deformation and Willmore energy in CV; in section \ref{sec3}, we derive the deformation-induced corrections to generalized holographic complexity and discuss their geometric interpretation. In section \ref{sec4}, we reformulate the results in terms of Lorentzian threads and show how multiple-flavor thread configurations arise naturally. We conclude in section \ref{con} with a discussion of the implications of our results and possible future directions.

\section{Generalized complexity at finite cut-off}\label{sec2}
In this section we briefly review the finite cut-off geometry arising from the generalized root-$T\bar T$ deformation, following~\cite{faraj}. The resulting Fefferman-Graham expansion provides the geometric framework used throughout the remainder of this work. We summarize only the ingredients required for the derivation of the generalized holographic complexity functional presented in the next section.

The holographic dual of $T\bar T$ deformation theories is defined at a finite cut-off $r=r_c$ \footnote{Recently, it has been proposed that instead of duality of the bulk at a finite cut-off and CFT deformation, there is an alternative geometrical duality where the bulk theory at the boundary is dual to an (undeformed CFT)+(a massive gravity) \cite{modhol}.}, where $T$ is the energy-momentum tensor of the boundary theory \cite{modhol}
\begin{align}
S_{CFT}+\Gamma \int \sqrt{h} \sqrt{\frac{1}{2}T_{\mu\nu}T^{\mu\nu}-\frac{1}{4}T_{\mu}^{\mu}T_{\nu}^{\nu}}.
\end{align}
Recently, the holographic complexity of $T\bar{T}$-deformation has been computed in the CV and CA proposals where the boundary theory lives on a finite cut-off radius in the bulk \cite{faraj}. It is worth mentioning that our work considers the complexity of the bulk theory via "complexity=anything" proposal. 

Due to the presence of the energy tensor in the Lagrangian, the non-locality of $T(x,x')$ enters into the theory and as a consequence, it suggests the larger $\Gamma$, the parameter of the deformation, the larger non-locality of the boundary theory. This non-locality could be represented by the entangled quantum gates in quantum circuits where the far apart parts of the circuit be entangled. The generalized $T\bar T$ deformation modifies the holographic dictionary by introducing a finite radial cut-off. Consequently, bulk geometric observables become functions of the cut-off surface and admit an expansion in powers of the cut-off parameter.

The complexity functional \ref{maxv} for $\alpha=0$ (CV proposal) in FG coordinates at the finite cut-off $\rho_c$, and for its analogues at the boundary (the renormalized original theory) read respectively as \cite{faraj}
\begin{align}
\mathcal{C}^{T\bar{T}}_{V}&=\frac{1}{G_N}\int d^{d-1}y\int_{\rho_c}d\rho \sqrt{H},\label{vexp}\\
&=\frac{\rho_c^{\frac{1-d}{2}}}{(d-1)G_N}\Bigg(V_{\Sigma}-\frac{(d-1)\rho_c}{2(d-3)(d-2)}\int_{\Sigma}d^{d-1}y\sqrt{h}\Big(G_{nn}-\frac{d-2}{(d-1)^2}K^2\Big)\Bigg),\\
\mathcal{C}^{undeformed}_{V}&=\frac{\rho_c^{\frac{1-d}{2}}}{(d-1)G_N}\Bigg(V_{\Sigma}-\frac{(d-1)\rho_c}{2(d-3)(d-2)}\int_{\Sigma}d^{d-1}y\sqrt{h}\Big(G_{nn}-\frac{(d-2)^2}{(d-1)^2}K^2\Big)\Bigg),\label{undef}
\end{align}
where $\sqrt{H}$ the determinant of the hypersurface, $y_a$ coordinates on the hypersurface $\Sigma$ at the finite cut-off, $\rho_c=\frac{1}{r_c}$, $G_{nn}$ the normal-normal component of the Einstein tensor and $K=g_{\mu\nu}K^{\mu\nu}$ the trace of the extrinsic curvature. One could evaluate how it is difficult to deform a given quantum state by $T\bar{T}$ deformation as follows \cite{faraj}
\begin{align}\label{subtract}
\delta \mathcal{C}_V&=\mathcal{C}^{undeformed}_V-\mathcal{C}^{T\bar{T}}_V=\frac{\rho_c^{\frac{3-d}{2}}}{G_N}W_{\Sigma},
\end{align}
where $W_{\Sigma}$ is the Willmore energy, which is the bending energy of hypersurface $\Sigma$:
\begin{align}
W_{\Sigma}=\frac{1}{2(d-1)^2}\int_{\Sigma} d^{d-1}y \sqrt{h}K^2.
\end{align}
It should be noted that by Gauss-Codazzi identity, complexities have been computed could be written entirely in terms of the extrinsic curvature of the hypersurface and its Euler number \cite{faraj}. 

The results summarized above follow directly from \cite{faraj} and serve as the geometric input for our analysis. The finite cut-off geometry reviewed above provides the essential ingredients for evaluating generalized geometric functionals. In the next section we employ this framework to derive the leading corrections to holographic complexity within the generalized ``complexity = anything'' proposal.
\section{Generalized holographic complexity under $T\bar{T}$ deformation} \label{sec3}
We now extend the analysis of the CV proposal to the generalized ``complexity = anything'' framework. The goal is to determine how the finite cut-off $T\bar{T}$ deformation modifies a generic geometric complexity functional and to identify the corresponding leading corrections.

The complexity functional \ref{maxv} as a generalization of \ref{vexp} would be \footnote{It should be noted that in the following definition a  maximization procedure is implied, however, we do not write it unless it poses a confusion.}
\begin{align}\label{comdef}
\mathcal{C}^{T\bar{T}}_{Any}=\frac{1}{G_N}\int d^{d-1}y\int_{\rho_c}d\rho \sqrt{H} F_1.
\end{align}
The generalized complexity proposal allows the scalar density $F_1$ to be any local curvature invariant. Typical examples include
\begin{align}
&F_{C^2}=1+\alpha L^4C^2,\\
&F_R=1+\alpha L^4 R^2,\\
&F_{R^2}=1+\alpha L^4 R_{\mu\nu\rho\sigma}R^{\mu\nu\rho\sigma},
\end{align}
where $C^2$ is the Weyl squared, $R$ the Ricci scalar and $R^2$ the Riemann tensor squared. 

To evaluate the generalized complexity functional perturbatively, we now expand the bulk metric and curvature invariants near the cut-off surface using FG coordinates \cite{epjp2026,faraj}
\begin{align}
&\sqrt{H}=\frac{1}{2\rho^{\frac{d+1}{2}}}\Big(1+\frac{\rho}{2(d-2)} G_{nn}+(\rho-\rho_c)\frac{K^2}{2(d-1)}+...\Big),\\
&R_{\mu\nu}R^{\mu\nu}=d^2(d+1)+\rho A+\mathcal{O}(\rho^2),\\
&R_{\mu\nu\rho\sigma}R^{\mu\nu\rho\sigma}=2d(d+1)+\rho  B+\mathcal{O}(\rho^2),\\
&R=-\frac{d(d+1)}{l^2}+\rho R(g^{(0)})+...,\\
&C^2=\rho^2 \mathscr{C}+\rho^2Log\rho \mathcal{A}+\mathcal{O}(\rho^3),
\end{align}
As a matter of convenience, we denote the expansion of every curvature invariant (CI) as follows
\begin{align}
CI=M_0+\rho M_1+\rho^2 M_2+... .
\end{align}
Then, the complexity functional \ref{comdef} admits the following expansion \footnote{We set $L,l=1$ in the rest of the paper.}
\begin{align}\label{comdefex}
&\mathcal{C}^{T\bar{T}}_{Any}=\frac{1}{G_N}\int d^{d-1}y\int_{\rho_c}d\rho \frac{1}{2\rho^{\frac{d+1}{2}}}\Big(1+\frac{\rho}{2(d-2)} G_{nn}\nonumber\\
&+(\rho-\rho_c)\frac{K^2}{2(d-1)}+...\Big) \Big(1+\alpha (M_0+\rho M_1+\mathcal{O}(\rho^2))\Big),\\
&=\frac{\alpha\rho_c^{\frac{1-d}{2}}}{(d-1)G_N}\Bigg(\int_{\Sigma}d^{d-1}y\sqrt{h}(\frac{1}{\alpha}+M_0)-\frac{(d-1)\rho_c}{2(d-3)(d-2)}\int_{\Sigma}d^{d-1}y\sqrt{h}(\frac{1}{\alpha}+M_0)\Big(G_{nn}\nonumber\\
&-\frac{d-2}{(d-1)^2}K^2\Big)\Bigg)+\frac{\alpha\rho_c^{\frac{3-d}{2}}}{(d-3)G_N}\Bigg(\int_{\Sigma}d^{d-1}y\sqrt{h}M_1-\frac{(d-1)\rho_c}{2(d-5)(d-2)}\int_{\Sigma}d^{d-1}y\sqrt{h}M_1\Big(G_{nn}\nonumber\\
&-\frac{d-2}{(d-1)^2}K^2\Big)\Bigg),
\end{align}
and the complexity functional for the undeformed would be \cite{faraj, comment}
\begin{align}
&\mathcal{C}^{undeformed}_{Any}=\frac{1}{G_N}\int d^{d-1}y\int_{\rho_c} d\rho \frac{1}{2\rho^{\frac{d+1}{2}}}\Big(1+\frac{\rho}{2(d-2)} G_{nn}\nonumber\\
&+(\rho-\rho_c)\frac{(d-2) K^2}{2(d-1)}+...\Big) \Big(1+\alpha (M_0+\rho M_1+\mathcal{O}(\rho^2))\Big),\\
&=\frac{\alpha \rho_c^{\frac{1-d}{2}}}{(d-1)G_N}\Bigg(\int_{\Sigma}d^{d-1}y\sqrt{h}(\frac{1}{\alpha}+M_0)-\frac{(d-1)\rho_c}{2(d-3)(d-2)}\int_{\Sigma}d^{d-1}y\sqrt{h}(\frac{1}{\alpha}+M_0)\Big(G_{nn}\nonumber\\
&-\frac{(d-2)^2}{(d-1)^2}K^2\Big)\Bigg)+\frac{\alpha \rho_c^{\frac{3-d}{2}}}{(d-3)G_N}\Bigg(\int_{\Sigma}d^{d-1}y\sqrt{h}M_1-\frac{(d-1)\rho_c}{2(d-5)(d-2)}\int_{\Sigma}d^{d-1}y\sqrt{h}M_1\Big(G_{nn}\nonumber\\
&-\frac{(d-2)^2}{(d-1)^2}K^2\Big)\Bigg).
\end{align}
We have introduced the coefficients $M_0$ and $M_1$ through the expansion $F_1=M_0+\rho M_1$ which collect the leading and next-to-leading curvature contributions. Throughout this work we retain terms up to $\mathcal{O}(\rho)$, corresponding to the leading correction induced by the finite cut-off deformation. Higher-order terms are suppressed in the small cut-off regime $\rho_c\ll 1$. The same expansion could be written for the undeformed complexity \ref{undef}. Similar to the complexity of the deformation in \ref{subtract}, the complexity of the deformation in "complexity=anything" could be written as follows
\begin{align}
\delta \mathcal{C}_{Any}&=\mathcal{C}_{Any}^{undeformed}-\mathcal{C}_{Any}^{T\bar{T}}\\
&=\frac{\rho_c^{\frac{3-d}{2}}}{G_N}W_{\Sigma}^{0-cor.}+\frac{\rho_c^{\frac{5-d}{2}}}{G_N}W_{\Sigma}^{1-cor.},
\end{align}
where $W_{\Sigma}^{0-cor.}$ is the corrected Willmore energy at the first order and $W_{\Sigma}^{1-cor.}$ at the second order
\begin{align}
&W_{\Sigma}^{0-cor.}=\frac{\alpha}{2(d-1)^2}\int _{\Sigma} d^{d-1}y\sqrt{h}(\frac{1}{\alpha}+ M_0)K^2,\\
&W_{\Sigma}^{1-cor.}=\frac{\alpha}{2(d-1)(d-5)}\int _{\Sigma} d^{d-1}y\sqrt{h} M_1 K^2.
\end{align}
In general one can generalized the complexity of the formation as follows
\begin{align}\label{sectors}
\delta C_{Any}=\frac{1}{G_N}\sum_{n=0}^{N} \rho_c^{\frac{3+2n-d}{2}}W_{\Sigma}^{(n)},
\end{align}
where the generalized Willmore functionals are defined by
\begin{align}
&W_{\Sigma}^{(n)}=\frac{\alpha}{2(d-1)(d-1-4n)}\sum_I\int _{\Sigma} d^{d-1}y\sqrt{h}M_n^I K^2,\\
&F_1=\sum_{n,I} \rho^n M_n^I,
\end{align}
where the $I$-index labels different curvature invariants, $\{F_I: R^2, R_{\mu\nu}R^{\mu\nu}, ...\}$. The equation \eqref{sectors} constitutes the main result of this section. It shows that the deformation-induced correction organizes naturally into a hierarchy of generalized Willmore-type functionals, each weighted by a definite power of the cut-off parameter. The expansion \eqref{sectors} shows that the finite cut-off deformation modifies the generalized complexity through an infinite sequence of generalized Willmore-type functionals. Since each coefficient $M_n$ originates from the curvature dependence of $F_1$, the correction naturally decomposes into independent geometric sectors. This observation provides the starting point for the Lorentzian-thread interpretation developed in the following section.
\section{Interpretation in terms of Lorentzian threads}\label{sec4}

The generalized complexity functional derived in the previous section was obtained independently of a Lorentzian-thread description. However, its structure inherently allows for such an interpretation, which provides a geometric realization of the generalized complexity functional and indicates a natural connection with the multiple-flavor Lorentzian-thread framework proposed in~\cite{pedraza1}.

Within the conventional complexity-volume (CV) proposal, Lorentzian threads offer a dual description of the maximal bulk hypersurface via the Lorentzian min flow-max cut theorem~\cite{minmax,pedraza2}. The generalized complexity functional examined in this work differs from the standard CV proposal by incorporating curvature-dependent contributions. Because these contributions enter linearly in the generalized volume functional, it is pertinent to investigate whether the associated Lorentzian flow can be decomposed into independent sectors corresponding to individual curvature invariants.

This section pursues two objectives. First, the generalized complexity functional is reformulated in terms of Lorentzian threads after an appropriate conformal rescaling of the bulk geometry. Second, it is demonstrated that the linear structure of the generalized functional motivates a local decomposition of the corresponding flow into multiple sectors, thus providing a geometric realization of the multi-flavor thread framework. Nonetheless, the conditions under which there is a global decomposition remain open and will be discussed below.

\subsection{Lorentzian-thread formulation of the generalized complexity functional}

Let $\mathcal{M}$ be a compact, oriented $(d+1)$-dimensional Lorentzian manifold with boundary $\partial\mathcal{M}$. Let $A\subset\partial\mathcal{M}$ denote a boundary subregion with entangling surface
\[
\partial A=\sigma_A .
\]
We denote by $\Sigma$ a bulk codimension-one hypersurface homologous to $A$, while $\widetilde{\Sigma}$ represents the maximal-volume hypersurface anchored on $\sigma_A$.

A Lorentzian flow is a future-directed vector field
$v^\mu$ satisfying

\begin{align}
\nabla_\mu v^\mu=0,
\qquad
|v|\ge1,
\qquad
v^0>0.
\end{align}

The integral curves of $v^\mu$ define the Lorentzian bit threads. The Lorentzian min flow-max cut theorem states that the minimum flux of such a flow through the boundary region $A$ is equal to the maximal volume of $\widetilde{\Sigma}$~\cite{minmax}. Consequently, the CV complexity may be written equivalently as

\begin{align}
\mathcal{C}_{\rm CV}
=
\frac{1}{G_NL}
\min_{v}
\int_A
n_\mu v^\mu
=
\frac{1}{G_NL}
\max_{\Sigma\sim A}
\int_\Sigma
d^{\,d}\sigma\,
\sqrt{h},
\end{align}

where $n^\mu$ denotes the future-directed unit normal to the boundary hypersurface and $h$ is the determinant of the induced metric on $\Sigma$.

The generalized complexity functional takes the form

\begin{align}
\mathcal{C}_{\rm Any}
=
\frac{V_x}{G_NL}
\max_{\Sigma}
\int_\Sigma
d^{\,d}\sigma\,
\sqrt{h}\,
F_1,
\end{align}

where the scalar function $F_1$ contains the curvature invariants.

We introduce the conformally rescaled metric

\begin{align}
\widetilde g_{\mu\nu}
=
F_1^{2/d}
g_{\mu\nu},
\end{align}

for which the induced volume element satisfies

\begin{align}
d^{\,d}\widetilde{\sigma}\,
\sqrt{\widetilde h}
=
F_1\,
d^{\,d}\sigma\,
\sqrt h .
\end{align}

The generalized complexity functional therefore becomes

\begin{align}
\label{minmax}
\mathcal{C}_{\rm Any}
=
\frac{V_x}{G_NL}
\max_{\Sigma}
\int_\Sigma
d^{\,d}\widetilde{\sigma}\,
\sqrt{\widetilde h}.
\end{align}

Equation~(\ref{minmax}) has precisely the same form as the standard CV functional, except that it is evaluated with respect to the conformally rescaled geometry. Consequently, the Lorentzian min flow-max cut theorem applies without further modification, leading to the equivalent flow representation

\begin{align}
\mathcal{C}_{\rm Any}
=
\min_{\widetilde v}
\int_A
\widetilde v ,
\end{align}

where the rescaled flow satisfies

\begin{align}
\widetilde\nabla_\mu
\widetilde v^\mu
=
0,
\qquad
|\widetilde v|
\ge
1.
\end{align}

This result establishes that the generalized complexity functional admits a Lorentzian-thread description after the conformal rescaling determined by the generalized volume density.
\subsection{Decomposition into multiple flavor sectors}
The multiflavor thread formalism defines a collection of independent thread species \cite{pedraza2,pedraza3}
\begin{align}
{v_I^\mu}, \qquad I=1, ...,N,
\end{align}
which simultaneously satisfy the conservation equations
\begin{align}
\nabla_\mu v_I^\mu = 0,
\end{align}
together with a norm constraint of the form
\begin{align}
|v_I| \ge1.
\end{align}
In contrast to the single-flow Lorentzian construction, each flavor in the multiflavor formalism represents an independent information channel. Consequently, the multiflavor articulates competently multiple cases where multiple information streams exist and interact, well-suited to multipartite information and distributed information processing \cite{pedraza2}.

Within this framework, non-local computation is characterized by the cooperative optimization of all thread species, rather than the propagation of a single flow. A computational task corresponds to a correlated configuration of the vector fields $v_I$, with the norm constraint dictating how the various channels allocate spacetime information. Because no individual thread species contains the complete information, computation is inherently distributed across the network of interacting flows. The multiflavor construction thus affords a more comprehensive geometric description of non-local computation by enabling the simultaneous optimization of multiple correlated information channels, clarifying how global computing operations arise via the interplay among several thread networks.

The generalized complexity functional \eqref{maxv} \footnote{It is worth mentioning that the functional is linear but the extrimization process in the complexity is not linear.}, is linear in the curvature contributions entering the scalar function $F_1$. Explicitly, one may write

\begin{align}
F_1=\sum_{I=1}^{N}a_I\,F_I,
\label{Fdecomp}
\end{align}
where the coefficients $a_I$ are determined by the underlying gravitational theory and the functions $F_I$ denote the independent geometric invariants contributing to the generalized volume functional. In the present construction the index $I$ labels the different curvature sectors, although the following discussion is independent of their explicit form.

The linearity of~(\ref{Fdecomp}) naturally suggests that the Lorentzian flow associated with the generalized complexity functional may also admit a decomposition into independent sectors. Such a decomposition would provide a geometric realization of the multiple-flavor Lorentzian-thread picture introduced in~\cite{pedraza1}, where different thread species are interpreted as carrying distinct microscopic information.

Motivated by this observation, we seek a family of divergence-free vector fields $v_I$, whose sum reproduces the generalized Lorentzian flow,
\begin{align}
v=\sum_{I=1}^{N}v_I.
\label{vsum}
\end{align}

Each vector field $v_I$ is interpreted as a separate thread flavor associated with one of the curvature invariants appearing in the eq.\ref{Fdecomp} \footnote{We should note that in this construction we do not mean every positive curvature invariants induces a Lorentzian flow, but we construct a linear map $\mathcal{L}$, that associate every positive curvature invariant to a Lorentzian thread $\mathcal{L}(F_I)=v_I$, which is not a unique correspondence.}.

The existence of such a decomposition is not automatic because the divergence-free condition must simultaneously be preserved for every sector. Nevertheless, as we now demonstrate, one may explicitly construct the required family of flows.

To this end we introduce the weights

\begin{align}
w_I
=
\frac{a_I F_I}{F_1},
\qquad
\sum_{I=1}^{N}w_I=1,
\label{weights}
\end{align}
which measure the local fractional contribution of each curvature invariant to the generalized density. A naive decomposition,

\begin{align}
v_I=w_Iv,
\end{align}
correctly reproduces the eq.~(\ref{vsum}), but in general fails to satisfy the divergence-free condition since

\begin{align}
\nabla_\mu(w_Iv^\mu)
=
v^\mu\nabla_\mu w_I .
\end{align}

To restore incompressibility we therefore introduce compensating vector fields $X_I$ satisfying \footnote{Eq. \eqref{split} is similar to perturbation of the Lorentzian threads in \cite{pedraza3}; however, we should note that \cite{pedraza3} perturbs the background metric and investigates the correction to the Lorentzian thread; in our work, the background metric is fixed and we are looking for a family of Lorentzian thread perturbations that yield $v=\sum v_I$.}

\begin{align}
v_I=(w_I+\lambda_I)v+X_I,
\label{split}
\end{align}
as Fig.(\ref{fig1}) shows, with the additional constraint

\begin{align}
\sum_{I=1}^{N}X_I=0.
\label{usum}
\end{align}

The remaining task is to determine the fields $X_I$ so that every sector obeys

\begin{align}
\nabla_\mu v_I^{\mu}=0.
\end{align}

This leads naturally to the following proposition.
\begin{figure}[ht]
\centering
\begin{overpic}[width=0.7\textwidth]{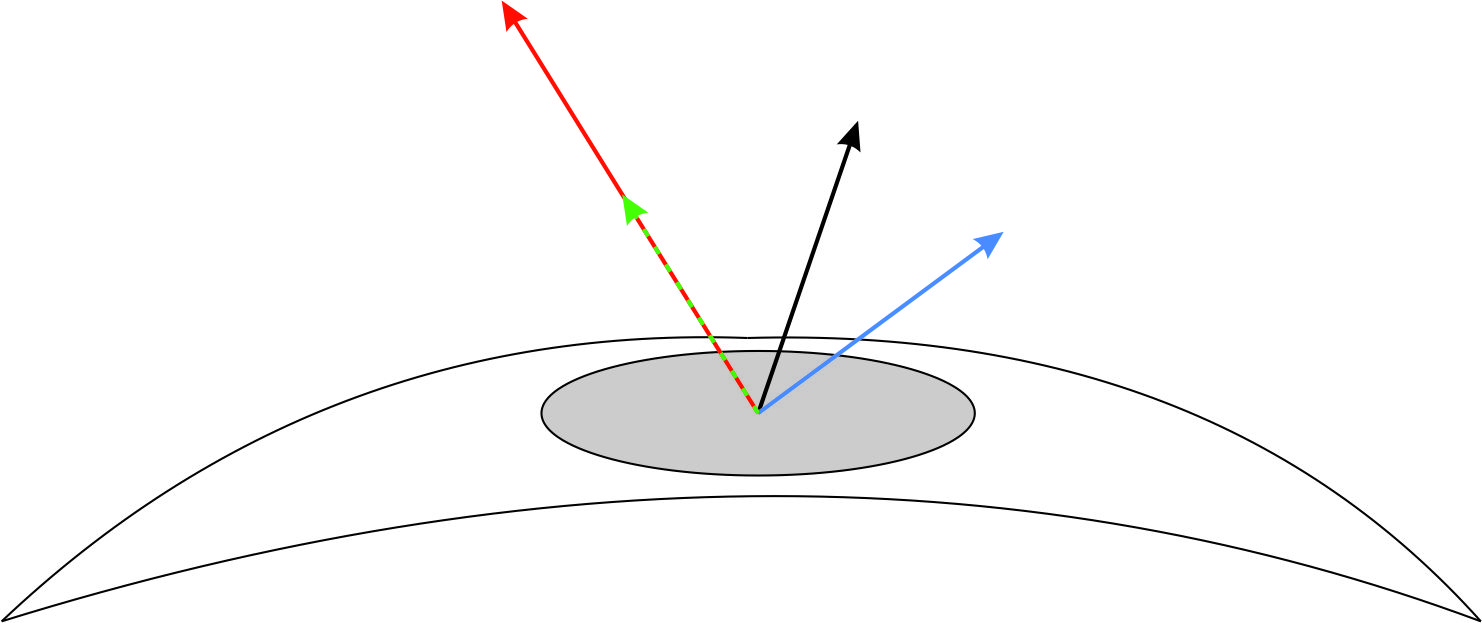}
\put(45,11){$p \in U$}
\put(65,29){\color{blue}$X_I \in v^{\perp}$}
\put(55,39){$v_I=(w_I+\lambda_I)v+X_I$}
\put(31,40){\color{red}$v$}
\put(22,25){\color{green}$(w_I+\lambda_I)v$}
\put(10,5){$\Sigma(\tau)$}
\end{overpic}
\caption{The local decomposition of the Lorentzian thread $v$ at $p\in U$.}\label{fig1}
\end{figure}
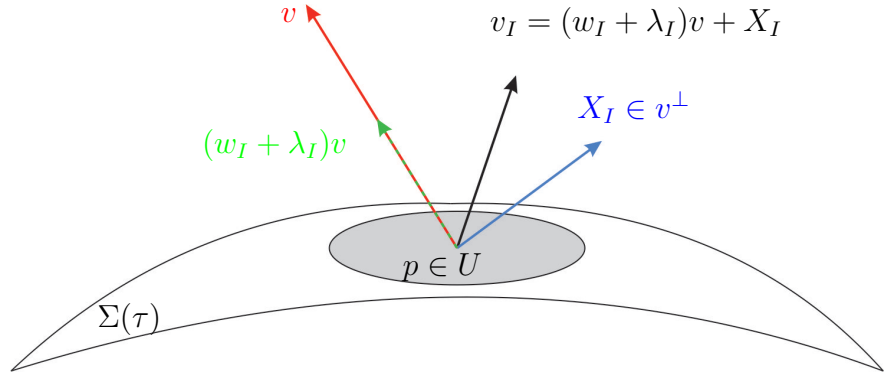
\subsection{Construction of the flavor flows}

The decomposition introduced above can be made precise by explicitly constructing a family of divergence-free vector fields satisfying the required local norm bounds.

\begin{prop}[Conditional global decomposition]\label{pro1}
Let $(M,g)$ be a Lorentzian manifold and let
$v\in\Gamma(T_M)$ be an optimal Lorentzian thread satisfying
\begin{equation}
\nabla_\mu v^\mu=0,
\qquad
|v|\ge 1 .
\end{equation}

Let
\begin{equation}
w_I=\frac{a_IF_I}{F},
\qquad
F=\sum_{I=1}^{N}a_IF_I,
\end{equation}
and let $\lambda_I$ be the solution of the equation
\begin{equation}\label{diffeq}
v^\mu\nabla_\mu(w_I+\lambda_I)=-\nabla X_I.
\end{equation}

Assume there exist defined smooth vector fields
$X_I\in\Gamma(T_M)$ satisfying
\begin{equation}
g(v,X_I)=0,
\qquad
\nabla_\mu X_I^\mu=-v\nabla(w_I+\lambda_I).
\end{equation}

Define
\begin{equation}
v_I=(w_I+\lambda_I)v+X_I.
\end{equation}

Then

\begin{enumerate}
\item each $v_I$ is divergence free,
\begin{equation}
\nabla_\mu v_I^\mu=0;
\end{equation}

\item if
\begin{equation}
\sum_{I=1}^{N}(w_I+\lambda_I)=f,
\qquad
\sum_{I=1}^{N}X_I=0,
\end{equation}
then
\begin{equation}\label{decom}
\sum_{I=1}^{N}v_I=f\,v;
\end{equation}

\item if
\begin{equation}
g(X_I,X_I)<(w_I+\lambda_I)^2|g(v,v)|,
\end{equation}
then every $v_I$ is timelike.
\end{enumerate}

\end{prop}

\begin{proof}

Using $\nabla_\mu v^\mu=0$, we obtain
\begin{align}
\nabla_\mu v_I^\mu
&=
\nabla_\mu\left((w_I+\lambda_I)v^\mu\right)
+\nabla_\mu X_I^\mu
\nonumber\\
&=
v^\mu\nabla_\mu(w_I+\lambda_I)
+(a_I+\lambda_I)\nabla_\mu v^\mu
+\nabla_\mu X_I^\mu
\nonumber\\
&=0.
\end{align}

The decomposition formula follows immediately from the assumptions
\(
\sum_i(a_i+\lambda_i)=f
\)
and
\(
\sum_iX_i=0.
\)

Finally,
\begin{equation}
g(v_I,v_I)
=
(w_I+\lambda_I)^2g(v,v)+g(X_I,X_I),
\end{equation}
since $g(v,X_I)=0$.
Hence the stated inequality guarantees
$g(v_I,v_I)<0$, so that each $v_I$ remains timelike.

\end{proof}
By choosing $w_I+\lambda_I=c_I$ where $c_I$ are constant then \eqref{diffeq} is automatically satisfied and as a result, $\sum c_I=f$ is constant, which means that \eqref{decom} is a rescaled decomposition of the Lorentzian flow $v$, where it could be normalized to unity. Furthermore, one could rescaled $X_I=\epsilon Y_I$, then by Proposition \ref{pro1}, $\nabla Y_I=0$. Consequently, $g(v_I,v_I)=(w_I+\lambda_I)^2|g(v,v)|+\epsilon^2g(Y_I,Y_I)$, by choosing
\begin{align}
\epsilon<\frac{(w_I+\lambda_I)\sqrt{|g(v,v)|}}{\sqrt{g(Y_I,Y_I)}},
\end{align}
the norm inequality is satisfied.

In the following we show that locally we could construct $X_I$ such that they are orthogonal to the Lorentzian flow $v$, satisfy $\sum X_I=0$.

\begin{prop}[Local orthogonal construction]\label{pro2}
Let $(M,g)$ be a Lorentzian manifold and let
$v$ be a smooth timelike Lorentzian thread.

Then, for every point $p\in M$, there exists a neighbourhood
$U\ni p$ together with a local orthonormal frame
\begin{equation}
\{e_A\}_{A=1}^{d-1}
\end{equation}
of the orthogonal bundle
\begin{equation}
v^\perp
=
\{X_I\in T_M\,:\,g(v,X_I)=0\},
\end{equation}
such that every orthogonal vector field admits the local decomposition
\begin{align}
&X_I=\sum_{A=1}^{d-1}(w_I-\bar w)\psi^{\,A}e_A ,\\
&\bar w=\frac{1}{N}\sum_I^N w_I=\frac{1}{N},
\end{align}
Moreover, the divergence-free condition
\begin{equation}
\nabla_\mu X_I^\mu=-v\nabla(w_I+\lambda_I)
\end{equation}
is equivalent to the following first-order linear system which is followed by a solution
\begin{equation}
\nabla_\mu\left((w_I-\bar w)\psi^{\,A}e_A^{\mu}\right)=v\nabla(w_I+\lambda_I).
\end{equation}

\end{prop}

\begin{proof}

Since $v$ is timelike, the orthogonal complement
$v^\perp$ is a smooth vector bundle of rank $d-1$.
Every smooth vector bundle is locally trivial; therefore there exists
a local orthonormal frame
$\{e_A\}_{A=1}^{d-1}$.

Consequently, every section
$X_I\in\Gamma(u^\perp)$
admits the unique local expansion
\begin{equation}
X_I=\sum_{A=1}^{d-1}(w_I-\bar w)\psi^{\,A}e_A.
\end{equation}

Substituting this decomposition into
$\nabla_\mu X_I^\mu=-v\nabla(w_I+\lambda_I)$
gives
\begin{equation}
\nabla_\mu\left((w_I-\bar w)\psi^{\,A}e_A^{\mu}\right)
=-v\nabla(w_I+\lambda_I),
\end{equation}
which is a linear first-order system for the scalar functions
$(w_I-\bar w)\psi^{\,A}e_A^{\mu}$, for which local solutions exist. Suppose a local coordinates $x^\mu$, then the divergence equation is written as follows
\begin{align}
&\nabla X_I=\frac{1}{\sqrt{|g|}}\partial_\mu \Big(\sqrt{|g|}X_I^\mu\Big)=h_I,\\
&h_I:=-v\nabla (w_I+\lambda_I).
\end{align}
Hence, the following first-order equation
\begin{align}
\partial_\mu \Big(\sqrt{|g|}X_I^\mu\Big)=\sqrt{|g|}h_I,
\end{align}
has a local solution by Poincare lemma.

In addition,
\begin{align}
\sum_I X_I=\sum_A \psi^Ae_A\sum_I (w_I-\bar w)=0
\end{align}

\end{proof}

The preceding proposition reduces the construction of the orthogonal
perturbations to the local solvability of the first-order system
\begin{equation}
\nabla_\mu\left((w_I-\bar w)\psi^{\,A}e_A^{\mu}\right)=v\nabla(w_I+\lambda_I).
\end{equation}

Therefore, a global existence theorem for the coefficients
$(w_I-\bar w)\psi^{\,A}e_A^{\mu}$ immediately yields globally defined orthogonal
divergence-free vector fields $X_I$, and hence, by Proposition~1,
a global multiflavor Lorentzian-thread decomposition
\begin{equation}
v_I=(w_I+\lambda_I)u+X_I.
\end{equation}

Establishing sufficient geometric or analytic conditions ensuring the
global solvability of the above system remains an interesting open
problem.

Thus, the problem of constructing the perturbations are reduced to solving a linear first-order system for the scalar coefficients $(w_I-\bar w)\psi^{,A}e_A$. Local solvability follows from the standard theory of linear first-order partial differential equations. The existence of global solutions depends on the global geometry of the orthogonal bundle $v^\perp$, the topology of the bulk spacetime, and the prescribed boundary conditions. To the best of our knowledge, a general global existence theorem for the system
\begin{align}\label{glob}
g(v,X_i)=0,
\qquad
\nabla_\mu X_I^\mu=-v\nabla(w_I+\lambda_I),
\end{align}
under the sole assumptions that $v$ is an optimal Lorentzian thread, is not currently available in the literature. Consequently, the existence of such global perturbations should be regarded as an additional assumption or established separately for the class of spacetimes under consideration. However, the global solution exists if the orthogonal distribution $v^\perp$ is tangent to a foliation by spacelike hypersurface as Fig.(\ref{fig2}) depicts.

From a geometric perspective, each flavor may be regarded as carrying the contribution of a particular curvature invariant to the generalized volume density. The total thread configuration is obtained by superposing all flavor sectors, while the generalized complexity is determined solely by their combined flux. This provides a natural realization of the multiple-flavor Lorentzian-thread picture proposed in ~\cite{pedraza1}.

All in all, we showed that by choosing $w_I+\lambda_I=c_I$, where $c_I$ are constant, there is a local decomposition of the Lorentzian thread $v_I=c_Iv+X_I$, associated to the generalized complexity functional containing $F_1=\sum a_IF_I$, where $\nabla v_I=0$ and $|v_I|\ge1$.

\begin{figure}[ht]
\centering
\begin{overpic}[width=0.6\textwidth]{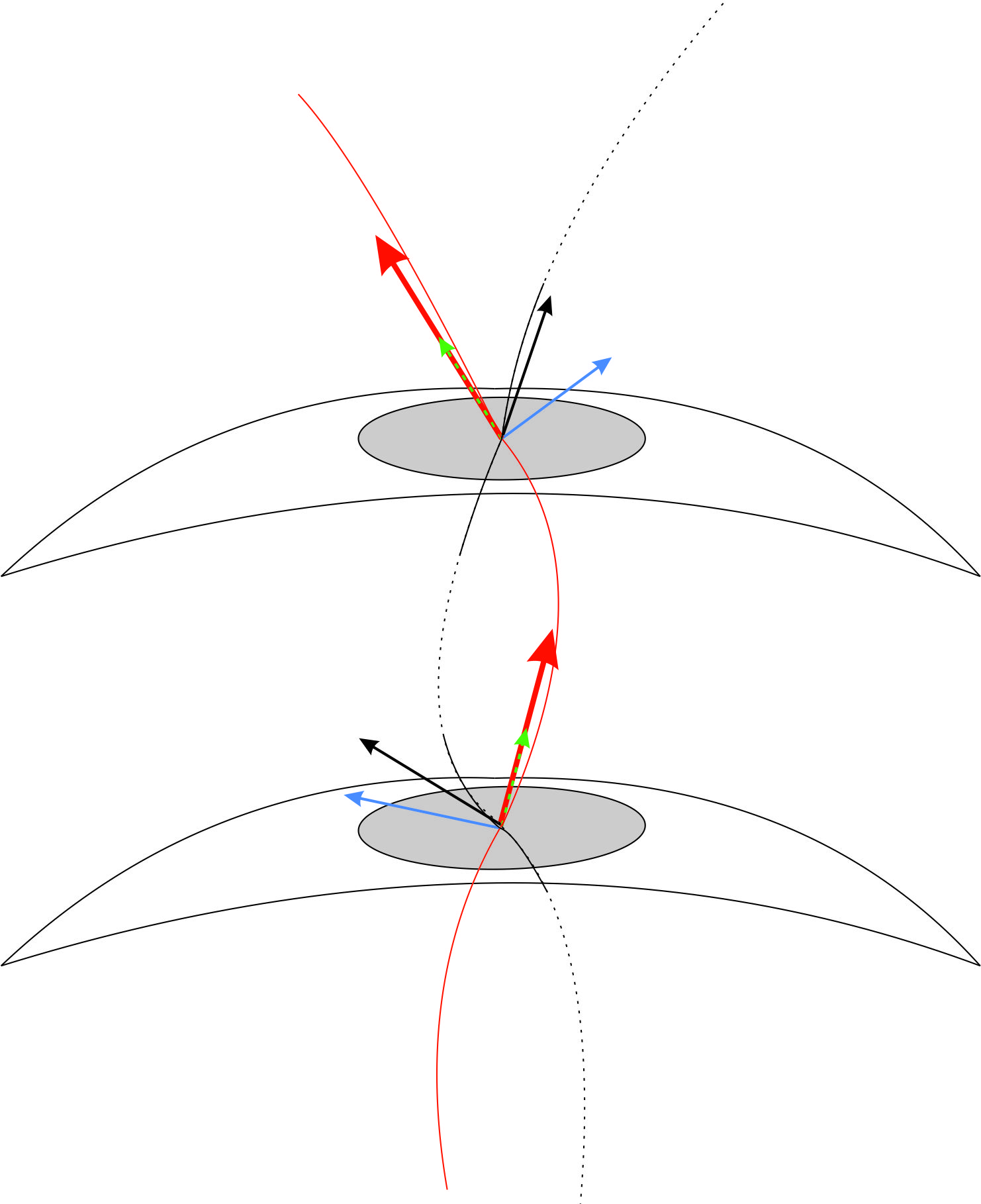}
\put(5,95){\color{red}$Lorentzian~thread~v$}
\put(58,90){$Global~Lorentzian~thread~v_I$}
\put(7,53){\fontsize{8}{10}\selectfont$Local~Lorentzian~thread~v_I$}
\put(0,29){$\Sigma(\tau_1)$}
\put(0,63){$\Sigma(\tau_2)$}
\end{overpic}
\caption{The local decomposition of the Lorentzian thread $v$ at two different boundary times $\tau$. The black dotted line shows the Lorentzian thread $v_I$, if there would be a global decomposition and the black solid line shows the local thread.}\label{fig2}
\end{figure}
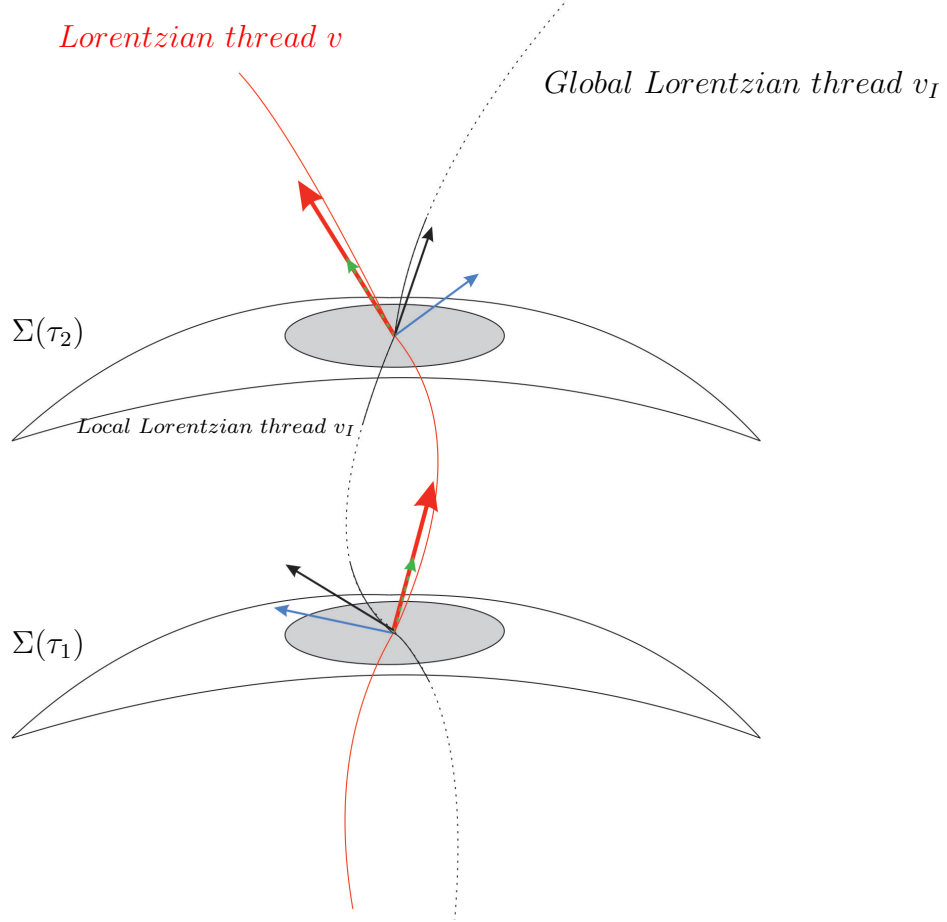
\subsection{A proposed interpretation in terms of quantum circuit layers}
The construction presented above provides a natural geometric interpretation of the generalized complexity functional in terms of multiple-flavor Lorentzian threads, in which the decomposition of the total flow is not unique and is locally defined throughout spacetime (or at least within the relevant causal domain). Restricting this vector field to a particular maximal slice $\tilde{\Sigma}$ gives the thread configuration at that instant $v_I^{\mu}|_{\tilde{\Sigma}}$. As the parameter $\tau$ changes, the slice changes, so the restriction of the same global vector field changes.

In the standard complexity-volume proposal, Lorentzian threads may be interpreted as representing elementary computational resources, with the total flux of these threads measuring the complexity of the boundary state~\cite{pedraza2,pedraza3}. The generalized complexity functional extends this picture by allowing several independent thread sectors to coexist. Since every flavor is associated with a different geometric contribution, it is natural to regard the individual sectors as describing distinct classes of elementary operations contributing to the preparation of the boundary state.

This interpretation becomes particularly transparent when one considers the foliation of the bulk spacetime by maximal hypersurfaces. Each hypersurface defines an instantaneous generalized complexity by extremizing the generalized volume functional, while the corresponding Lorentzian flow yields the associated thread configuration. As the hypersurface evolves in boundary time, one obtains a one-parameter family of generalized flows $v_I(\tau)$ \footnote{It appears that in the complexity=anything proposal there are some transitions between extremized hypersurfaces, which means that there is $\tau$ such that two hypersurfaces extremize the complexity, but one of them maximizes the complexity. In this way, we should denote the flows by $v(\tilde \Sigma)$ rather than $v(\tau)$.}, where the label $\tau$ denotes the boundary time and $I$ labels the thread flavor.

Because the $T\bar{T}$ deformation introduces nonlocality, it is expected that the complexity \eqref{sectors} shows multiple flavor threads reformulation. Eq. \eqref{sectors} at the leading order would be
\begin{align}
&\delta \mathcal{C}_{Any}\sim \rho_c^{\frac{3-d}{2}}W^{(0)},\label{lead}\\
&W^{(0)}\sim \int d^{d-1}y \sqrt{h} M_0K^2.
\end{align}
At each time $\tau$ there is a maximized-extrimized hypersurface $\Sigma$. If at each time we define $\tilde{h}_{ab}(\Sigma)=M_0^{\frac{1}{d-1}}K^{\frac{2}{d-1}}_{\Sigma}h_{ab}(\Sigma)$, where $K^2_{\Sigma}$ is the extrinsic curvature of the extrimized hypersurface at time $\tau$, then the complexity \eqref{sectors} at the leading order can be written as
\begin{align}
\delta \mathcal{C}_{Any}\sim \max \int d\tilde{V}.
\end{align}
where $d\tilde{V}=d^{d-1}y \rho_c^{\frac{3-d}{2}}\sqrt{h} M_0^IK^2$. As a consequence, the complexity \eqref{sectors}, by min flow-max cut theorem would be written in terms of Lorentzian threads
\begin{align}
\delta \mathcal{C}_{Any}(\Sigma)\sim \min_{\tilde{v}} \int \tilde{v}(\Sigma)
\end{align}

Now, one supposes that $F_1=\sum_I a_IF_I$, then \eqref{lead} would be
\begin{align}\label{comfor}
\delta \mathcal{C}_{Any}\sim \rho_c^{\frac{3-d}{2}}\sum_I  \int d^{d-1}y \sqrt{h}a_IM_0^I K^2,
\end{align}
where the FG expansion leads to $F_1=\sum\limits_{I,n}a_I M_n^I$. Then, by the min flow-max cut and rescaling
\begin{align}\label{resc} 
d\tilde{V}=d^{d-1}y\rho_c^{\frac{3-d}{2}}\sqrt{h}\sum_Ia_IM_0^IK^2,
\end{align} 
the leading order of complexity would be
\begin{align}
\delta \mathcal{C}_{Any}\sim \max \int d\tilde V.
\end{align}
Like the previous section, the complexity \eqref{comfor} by rescaling \eqref{resc} is mapped to a maximization volume problem, where the min flow-max cut theorem shows that there is a Lorentzian thread reformulation for it where
\begin{align}
\delta \mathcal{C}_{Any}\sim \max \int d\tilde{V}=\min \int v.
\end{align}
By defining a weight function
\begin{align}
w_I=\frac{a_IM_0^I}{\sum_Ja_JM_0^J},
\end{align}
where by definition $\sum_Iw_I=1$, and $v_I=(w_I+\lambda_I)v+X_I$, the proposition \ref{pro2} has been proved leads to local $v_I$ as a flavor.

The successive maximal hypersurfaces therefore define a sequence of generalized thread configurations connecting the initial and final boundary states. We heuristically interpret each hypersurface as an individual layer of the quantum circuit, with different thread flavors describing independent channels through which geometric information contributes within that layer.

The above expression should not be regarded as an alternative definition of the generalized complexity but rather as a schematic representation of its decomposition into geometric sectors.

The interpretation proposed here is related to the multiple-flavor Lorentzian thread framework introduced in~\cite{pedraza1}. In that construction, different thread species are associated with distinct microscopic resources contributing to the complexity of the boundary state. The present analysis presents a geometric realization of this idea. Instead of introducing the flavors as independent ingredients, they emerge inherently from decomposing the generalized volume density into its curvature contributions. For visual description see Fig.(\ref{fig3}).

From this perspective, the generalized complexity functional contains more refined geometric information than the conventional complexity-volume proposal. The standard CV prescription is recovered when the generalized density reduces to a constant, in which case all flavor sectors merge into a single Lorentzian flow. Curvature invariants, on the other hand, separate the total flow into several independent components, each encoding the contribution of a particular geometric invariant.

Although the correspondence within flavor sectors and elementary gate classes remains heuristic, it provides an intuitive physical interpretation of the generalized Lorentzian-thread construction developed in this work. Establishing a precise microscopic identification between the thread flavors and specific gate sets would require a detailed understanding of the boundary quantum circuit and lies beyond the scope of the present analysis. Nevertheless, the geometric framework developed here suggests that generalized complexity naturally admits a richer internal structure than that described by the conventional complexity-volume proposal.

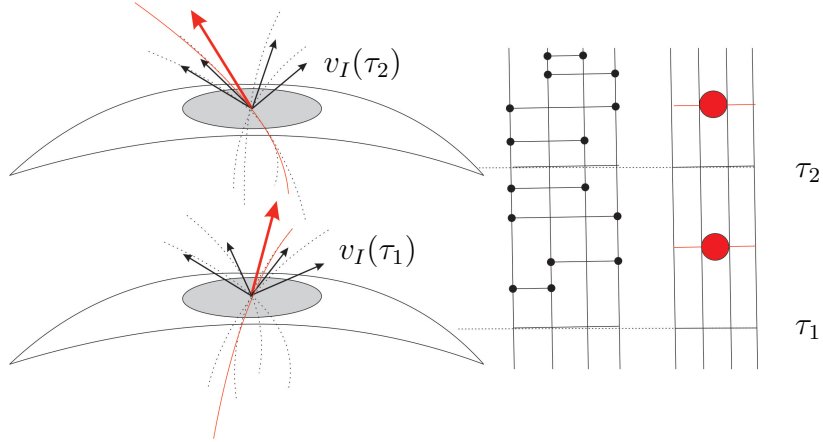
\begin{figure}[ht]
\centering
\begin{overpic}[width=0.6\textwidth]{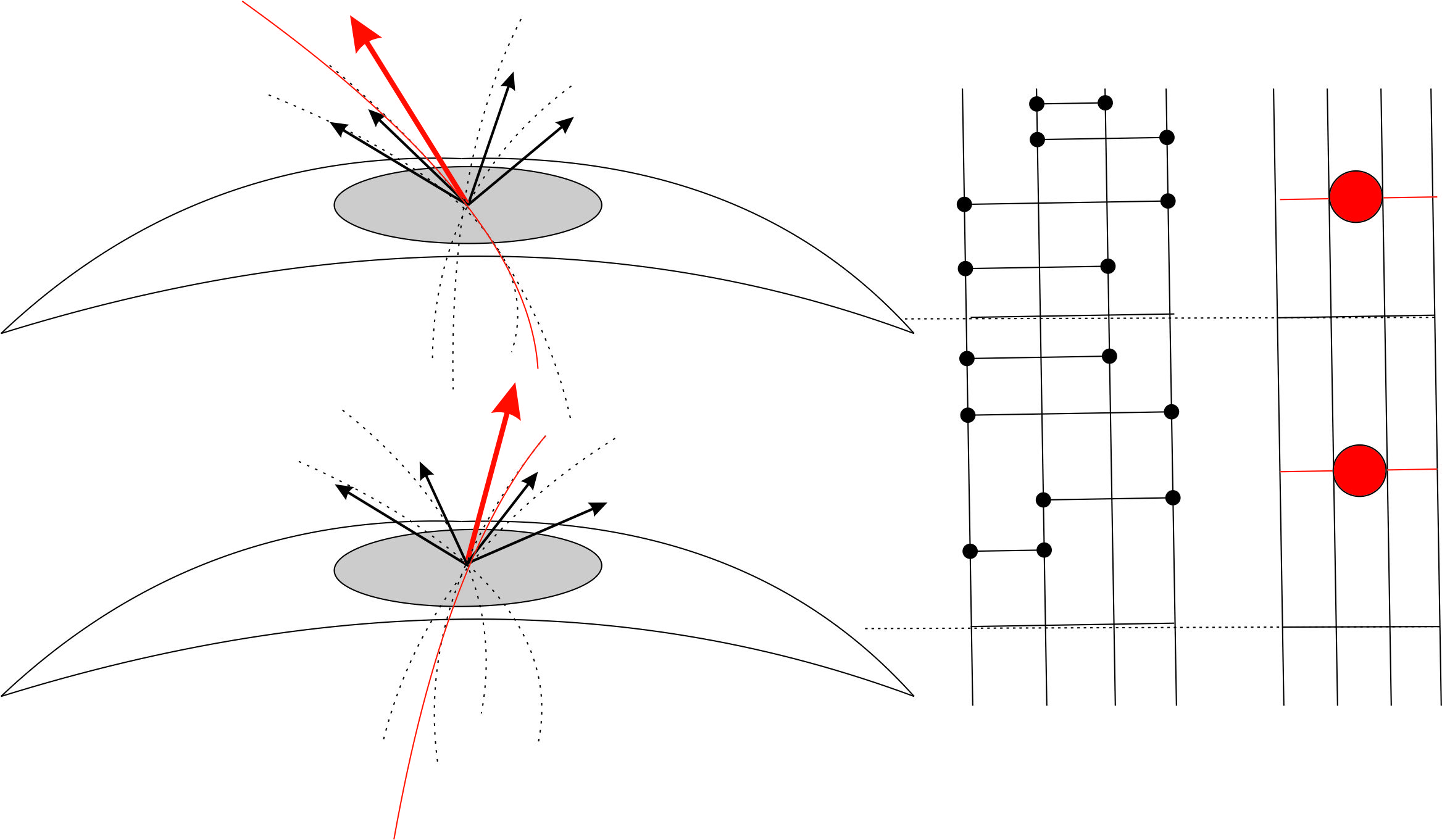}
\put(105,14){$\tau_1$}
\put(105,35){$\tau_2$}
\put(44,24){$v_I(\tau_1)$}
\put(42,49){$v_I(\tau_2)$}
\end{overpic}
\caption{The proposed interpretation suggests that at each hypersurface $\Sigma(\tau)$, it could be associated local multiple-Lorentzian threads as quantum gates as in \cite{pedraza2} in quantum circuit model.}\label{fig3}
\end{figure}

 \section{Conclusion}\label{con}
The impact of $T\bar T$-type deformations on generalized holographic complexity is investigated. Building on the proposed relationship between finite cut-off holography and generalized root-$T\bar T$ deformations, this analysis focuses on the complexity functional within the framework of the "complexity = anything" proposal and examines how the deformation modifies the associated geometric quantities.

A Fefferman-Graham expansion near the cut-off surface reveals the difference between deformed and undeformed complexities. The resulting correction is expressed as a series of generalized Willmore-type terms, with structure determined by the curvature invariants in the complexity functional. This approach offers a geometric characterization of how the deformation affects the complexity of the dual quantum state.

These corrections can be naturally interpreted in terms of Lorentzian threads. Specifically, generalized complexity decomposes into contributions from various geometric invariants, resulting in a framework with multiple thread types. In this context, the non-locality induced by the $T\bar T$ deformation appears as additional thread sectors, which may correspond to distinct classes of computational operations contributing to the preparation of the boundary state.

This analysis indicates a direct connection between non-local deformations, generalized holographic complexity, and the emerging thread-based description of bulk geometry. While the discussion remains primarily conceptual, it provides evidence that generalized complexity encodes information regarding the organization of non-local computational processes in the dual theory.

The present work demonstrates that generalized holographic complexity possesses an internal geometric structure encoded by curvature-dependent Lorentzian-thread sectors. Determining whether these sectors have a direct microscopic realization in the boundary theory remains an important direction for future research.

Several questions remain open. Further investigation of explicit finite cut-off gravitational backgrounds is needed to determine how thread decomposition depends on the deformation parameter. Exploring connections with other information-theoretic observables and developing a more microscopic understanding of the relationship between generalized complexity and non-local quantum circuits are also important. Additionally, proving a global existence theorem and identifying the conditions under which \eqref{glob} holds remain open problems.

\section*{Acknowledgment}
Authors extend their gratitude to Juan Pedraza for reading the manuscript and providing useful comments.


\end{document}